\documentclass[sigconf]{acmart}
\AtBeginDocument{%
  }

\copyrightyear{2025}
\acmYear{2025}
\setcopyright{cc}
\setcctype{by}
\acmConference[KDD '25]{Proceedings of the 31st ACM SIGKDD Conference on Knowledge Discovery and Data Mining V.2}{August 3--7, 2025}{Toronto, ON, Canada}
\acmBooktitle{Proceedings of the 31st ACM SIGKDD Conference on Knowledge Discovery and Data Mining V.2 (KDD '25), August 3--7, 2025, Toronto, ON, Canada}
\acmDOI{10.1145/3711896.3736950}
\acmISBN{979-8-4007-1454-2/2025/08}

\usepackage{hyperref}
\usepackage{amsmath}
\usepackage{soul}
\usepackage{graphicx}
\usepackage{algorithm}
\usepackage[noend]{algpseudocode}
\usepackage{natbib}
\usepackage{xspace}
\usepackage{subfig}
\usepackage{bbm}
\usepackage{bm}
\usepackage{centernot}
\usepackage{tcolorbox}
\usepackage{tabularx}

\usepackage{tikz}
\usepackage{pgfplots}

\definecolor{yafcolor1}{rgb}{0.4, 0.165, 0.553}
\definecolor{yafcolor2}{rgb}{0.949, 0.482, 0.216}
\definecolor{yafcolor3}{rgb}{0.47, 0.549, 0.306}
\definecolor{yafcolor4}{rgb}{0.925, 0.165, 0.224}
\definecolor{yafcolor5}{rgb}{0.141, 0.345, 0.643}
\definecolor{yafcolor6}{rgb}{0.965, 0.933, 0.267}
\definecolor{yafcolor7}{rgb}{0.627, 0.118, 0.165}
\definecolor{yafcolor8}{rgb}{0.878, 0.475, 0.686}

\usepackage{algorithm}
\usepackage[noend]{algpseudocode}

\newcommand{\bigO}{\ensuremath{\mathcal{O}}\xspace}
\newcommand{\NP}{\ensuremath{\mathbf{NP}}\xspace}

\algrenewcommand\algorithmicrequire{\textbf{Input:}}
\algrenewcommand\algorithmicensure{\textbf{Output:}}

\algblockdefx{MRepeat}{EndRepeat}{\textbf{repeat}}{}
\algnotext{EndRepeat}

\usepackage{array}
\usepackage{arydshln}
\newcommand{\ptitle}[1]{\smallskip\noindent{\bf #1.}}

\newcommand{\set}[1]{\left\{#1\right\}}
\newcommand{\pr}[1]{\left(#1\right)}
\newcommand{\fpr}[1]{\mathopen{}\left(#1\right)}

\newcommand{\abs}[1]{{\left|#1\right|}}
\newcommand{\floor}[1]{{\left\lfloor#1\right\rfloor}}

\newcommand{\enset}[2]{\left\{#1 ,\ldots , #2\right\}}

\newcommand{\define}{\leftarrow}

\newcommand{\funcdef}[3]{{#1}:{#2} \to {#3}}

\DeclareRobustCommand{\dispfunc}[2]{%
    \ensuremath{%
        \ifthenelse{\equal{#2}{}}%
            {\mathit{#1}}%
            {\mathit{#1}\fpr{#2}}}}

\newcommand{\diver}[1]{\dispfunc{div}{#1}}

\newcommand{\algpre}[1]{\dispfunc{\textsc{Preprocess}}{#1}}
\newcommand{\algdec}[1]{\dispfunc{\textsc{Decompose}}{#1}}
\newcommand{\algflow}[1]{\dispfunc{\textsc{Assign}}{#1}}
\newcommand{\algmmd}[1]{\dispfunc{\textsc{Breach}}{#1}}

\newcommand{\labflow}{\textsc{f-flow}}
\newcommand{\labilp}{\textsc{core}}
\newcommand{\labgrd}{\textsc{g-flow}}
\newcommand{\labmmd}{\textsc{Breach}}

\newcommand{\fairflow}[1]{\dispfunc{\textsc{FairFlow}}{#1}}
\newcommand{\greedyflow}[1]{\dispfunc{\textsc{GreedyFlow}}{#1}}
\newcommand{\silp}[1]{\dispfunc{\textsc{Core}}{#1}}

\newcommand{\FMMD}{\textsf{FMMD}\xspace}
\newcommand{\OPT}{\text{\normalfont OPT}\xspace}
\newcommand{\GMM}{\textsf{GMM}\xspace}

\newtheorem{theorem}{Theorem}
\newtheorem{lemma}{Lemma}

\newtheorem{prop}{Proposition}
\newtheorem{prob}{Problem}

\newlength{\figwidth}
\newlength{\figheight}

\pgfdeclarelayer{background}
\pgfdeclarelayer{foreground}
\pgfsetlayers{background,main,foreground}

\definecolor{yafaxiscolor}{rgb}{0.3, 0.3, 0.3}

\newlength{\yafaxispad}
\newlength{\yaftlpad}
\newlength{\yaflabelpad}
\newlength{\yafaxiswidth}
\newlength{\yafticklen}
\makeatletter
\def\pgfplots@drawtickgridlines@INSTALLCLIP@onorientedsurf#1{}
\makeatother

\newcommand{\yafdrawaxis}[4]{
	\pgfplotstransformcoordinatex{#1}\let\xmincoord=\pgfmathresult 
	\pgfplotstransformcoordinatex{#2}\let\xmaxcoord=\pgfmathresult 
	\pgfplotstransformcoordinatey{#3}\let\ymincoord=\pgfmathresult 
	\pgfplotstransformcoordinatey{#4}\let\ymaxcoord=\pgfmathresult 
	\pgfsetlinewidth{\yafaxiswidth} 
	\pgfsetcolor{yafaxiscolor}
	\pgfpathmoveto{\pgfpointadd{\pgfpointadd{\pgfplotspointrelaxisxy{0}{0}}{\pgfqpointxy{\xmincoord}{0}}}{\pgfqpoint{-0.5\yafaxiswidth}{\yafaxispad}}}
	\pgfpathlineto{\pgfpointadd{\pgfpointadd{\pgfplotspointrelaxisxy{0}{0}}{\pgfqpointxy{\xmaxcoord}{0}}}{\pgfqpoint{0.5\yafaxiswidth}{\yafaxispad}}}
	\pgfpathmoveto{\pgfpointadd{\pgfpointadd{\pgfplotspointrelaxisxy{0}{0}}{\pgfqpointxy{0}{\ymincoord}}}{\pgfqpoint{\yafaxispad}{-0.5\yafaxiswidth}}}
	\pgfpathlineto{\pgfpointadd{\pgfpointadd{\pgfplotspointrelaxisxy{0}{0}}{\pgfqpointxy{0}{\ymaxcoord}}}{\pgfqpoint{\yafaxispad}{0.5\yafaxiswidth}}}
	\pgfusepath{stroke}
}

\pgfplotscreateplotcyclelist{yaf}{%
{yafcolor5,mark options={scale=0.75},mark=o}, 
{yafcolor2,mark options={scale=0.75},mark=square},
{yafcolor3,mark options={scale=0.75},mark=triangle},
{yafcolor4,mark options={scale=0.75},mark=o},
{yafcolor1,mark options={scale=0.75},mark=o},
{yafcolor8,mark options={scale=0.75},mark=o},
{yafcolor6,mark options={scale=0.75},mark=o},
{yafcolor7,mark options={scale=0.75},mark=o}} 

\pgfkeys{/pgf/number format/.cd,1000 sep={\,}}

\pgfplotsset{axis y line=left, axis x line=bottom,
	tick align=outside,
	tickwidth=\yafticklen,
	clip = false,
    x axis line style= {-, line width = 0pt, color=black!0},
    y axis line style= {-, line width = 0pt, color=black!0},
    x tick style= {line width = \yafaxiswidth, color=yafaxiscolor, yshift = \yafaxispad},
    y tick style= {line width = \yafaxiswidth, color=yafaxiscolor, xshift = \yafaxispad},
    x tick label style = {font=\small, yshift = \yaftlpad, inner xsep = 0pt},
    y tick label style = {font=\small, xshift = \yaftlpad},
    every axis y label/.style = {at = {(ticklabel cs:0.5)}, rotate=90, anchor=center, font=\small, yshift = -\yaflabelpad, inner sep = 0pt},
    every axis x label/.style = {at = {(ticklabel cs:0.5)}, anchor=center, font=\small, yshift = \yaflabelpad},
    x tick label style = {font=\small, yshift = 1pt},
    title style={inner sep = 0pt, yshift = -4pt},
    grid = major,
    major grid style  = {dash pattern = on 1pt off 3 pt},
        every axis plot post/.append style= {line width=\yafaxiswidth} ,
	legend cell align = left,
	legend style = {inner sep = 1pt, cells = {font=\scriptsize}},
	legend image code/.code={%
		\draw[mark repeat=2,mark phase=2,#1] 
		plot coordinates { (0cm,0cm) (0.15cm,0cm) (0.3cm,0cm) };%
	} 
}

\begin{document}

\title{Fair Diversity Maximization with Few Representatives}

\author{Florian Adriaens}
\email{florian.adriaens@helsinki.fi}
\affiliation{%
  \institution{University of Helsinki, HIIT}
  \city{Helsinki}
  \country{Finland}
}

\author{Nikolaj Tatti}
\email{nikolaj.tatti@helsinki.fi}
\affiliation{%
  \institution{University of Helsinki, HIIT}
  \city{Helsinki}
  \country{Finland}
}

\renewcommand{\shortauthors}{Florian Adriaens, Nikolaj Tatti}

\begin{abstract}
Diversity maximization problem is a well-studied problem where the goal is to find $k$ diverse items.
Fair diversity maximization aims to select a diverse subset of $k$ items from a large dataset, while requiring that each group of items be well represented in the output.
More formally, given a set of items with labels,
our goal is to find $k$ items that maximize the minimum pairwise distance in the set, while maintaining that each label is represented within some budget.
In many cases, one is only interested in selecting a handful (say a constant) number of items from each group.
In such scenario we show that a randomized algorithm based on padded decompositions improves the state-of-the-art approximation ratio to $\sqrt{\log(m)}/(3m)$, where $m$ is the number of labels.
The algorithms work in several stages:
($i$) a preprocessing pruning which ensures that points with the same label are far away from each other,
($ii$) a decomposition phase, where points are randomly placed in clusters such that there is a feasible solution with maximum one point per cluster and that any feasible solution will be diverse,
$(iii)$ assignment phase, where clusters are assigned to labels, and a representative point with the corresponding label is selected from each cluster.
We experimentally verify the effectiveness of our algorithm on large datasets.
\end{abstract}


\maketitle

\newcommand\kddavailabilityurl{https://doi.org/10.5281/zenodo.15533799}

\ifdefempty{\kddavailabilityurl}{}{
\begingroup\small\noindent\raggedright\textbf{KDD Availability Link:}\\
The source code of this paper has been made publicly available at \url{\kddavailabilityurl}.
\endgroup
}


\section{Introduction}
\label{sec:intro}
Diversity maximization is an extremely well-studied task with applications in data summarization \cite{celis2018fair,CHANDRA2001438,zheng2017survey}, feature selection \cite{zadeh2017scalable}, facility location \cite{kuby1987programming,erkut1990discrete,tamir1991obnoxious,ravi1994heuristic}, recommendation systems \cite{abbassi2013diversity, kaminskas2016diversity, castells2021novelty}, web search \cite{xin2006extracting, radlinski2006improving,divtopk,bhattacharya2011consideration}, information exposure in social networks \cite{matakos2020tell} and e-commerce \cite{zheng2017survey}, among others.
The goal is to find a set of \emph{diverse} items within the data, which means that the selected items should be highly dissimilar to each other.

\begin{figure}
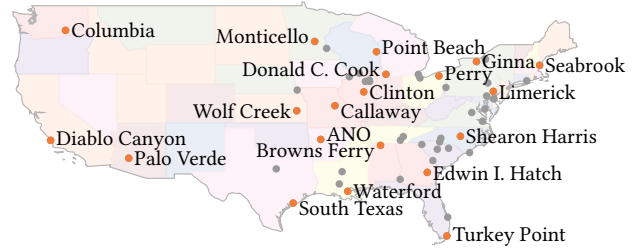

\begin{tikzpicture}[scale=0.13]
\input{states}
\input{coast}
\tikzstyle{marker} = [fill=black!40, circle, inner sep=1pt]

\draw (-93.23138889, 35.31027778) node[marker] {};
\draw (-80.43055556, 40.62333333) node[marker] {};
\draw (-88.22916667, 41.24361111) node[marker] {};
\draw (-87.11861111, 34.70388889) node[marker] {};
\draw (-78.01027778, 33.95833333) node[marker] {};
\draw (-89.28194444, 42.07416667) node[marker] {};
\draw (-91.78, 38.76166667) node[marker] {};
\draw (-76.44222222, 38.43194444) node[marker] {};
\draw (-81.07, 35.05166667) node[marker] {};
\draw (-88.835, 40.17222222) node[marker] {};
\draw (-119.3338889, 46.47111111) node[marker] {};
\draw (-97.785, 32.29833333) node[marker] {};
\draw (-95.64138889, 40.36194444) node[marker] {};
\draw (-83.08638889, 41.59666667) node[marker] {};
\draw (-120.8561111, 35.21083333) node[marker] {};
\draw (-86.56583333, 41.97527778) node[marker] {};
\draw (-88.26805556, 41.38972222) node[marker] {};
\draw (-82.34388889, 31.93416667) node[marker] {};
\draw (-83.2575, 41.96277778) node[marker] {};
\draw (-77.31, 43.27777778) node[marker] {};
\draw (-91.04833333, 32.00666667) node[marker] {};
\draw (-80.15833333, 34.40277778) node[marker] {};
\draw (-75.53805556, 39.46777778) node[marker] {};
\draw (-76.38583333, 43.51777778) node[marker] {};
\draw (-85.11166667, 31.22305556) node[marker] {};
\draw (-88.66916667, 41.24555556) node[marker] {};
\draw (-75.58722222, 40.22666667) node[marker] {};
\draw (-80.94833333, 35.4325) node[marker] {};
\draw (-72.16861111, 41.31194444) node[marker] {};
\draw (-93.84916667, 45.33361111) node[marker] {};
\draw (-76.40694444, 43.52083333) node[marker] {};
\draw (-77.78944444, 38.06055556) node[marker] {};
\draw (-82.89805556, 34.79388889) node[marker] {};
\draw (-112.865, 33.38916667) node[marker] {};
\draw (-76.26805556, 39.75833333) node[marker] {};
\draw (-81.14333333, 41.80083333) node[marker] {};
\draw (-87.53666667, 44.28111111) node[marker] {};
\draw (-92.63305556, 44.62166667) node[marker] {};
\draw (-90.31, 41.72638889) node[marker] {};
\draw (-91.33333333, 30.75111111) node[marker] {};
\draw (-80.24638889, 27.34861111) node[marker] {};
\draw (-75.53555556, 39.46277778) node[marker] {};
\draw (-70.85083333, 42.89888889) node[marker] {};
\draw (-85.09166667, 35.22638889) node[marker] {};
\draw (-78.95083333, 35.63333333) node[marker] {};
\draw (-96.04888889, 28.79555556) node[marker] {};
\draw (-76.69777778, 37.16555556) node[marker] {};
\draw (-76.14888889, 41.08888889) node[marker] {};
\draw (-80.33055556, 25.43416667) node[marker] {};
\draw (-81.31472222, 34.29861111) node[marker] {};
\draw (-81.76583333, 33.14305556) node[marker] {};
\draw (-90.47111111, 29.99527778) node[marker] {};
\draw (-84.78944444, 35.60277778) node[marker] {};
\draw (-95.68888889, 38.23888889) node[marker] {};

\tikzstyle{mark2} = [fill=yafcolor2, circle, inner sep=1pt]
\tikzstyle{label} = [anchor=west, font=\small, inner sep=2pt, inner ysep=-1pt]

\tikzstyle{mark2} = [fill=yafcolor2, circle, inner sep=1pt]
\tikzstyle{label} = [anchor=west, font=\small, inner sep=2pt, inner ysep=-1pt]

\draw (-120.8561111, 35.21083333) node[mark2] {} node[label] {Diablo Canyon};
\draw (-70.85083333, 42.89888889) node[mark2] {} node[label] {Seabrook};
\draw (-96.04888889, 28.79555556) node[mark2] {} node[label, anchor=north west] {South Texas};
\draw (-87.53666667, 44.28111111) node[mark2] {} node[label, anchor = south west] {Point Beach};
\draw (-80.33055556, 25.43416667) node[mark2] {} node[label] {Turkey Point};
\draw (-119.3338889, 46.47111111) node[mark2] {} node[label] {Columbia};
\draw (-78.95083333, 35.63333333) node[mark2] {} node[label] {Shearon Harris};
\draw (-95.68888889, 38.23888889) node[mark2] {} node[label, anchor=east] {Wolf Creek};
\draw (-87.11861111, 34.70388889) node[mark2] {} node[label, anchor=north east] {Browns Ferry};
\draw (-112.865, 33.38916667) node[mark2] {} node[label] {Palo Verde};
\draw (-81.14333333, 41.80083333) node[mark2] {} node[label] {Perry};
\draw (-93.84916667, 45.33361111) node[mark2] {} node[label, anchor=south east] {Monticello};
\draw (-90.47111111, 29.99527778) node[mark2] {} node[label] {Waterford};
\draw (-75.58722222, 40.22666667) node[mark2] {} node[label] {Limerick};
\draw (-82.34388889, 31.93416667) node[mark2] {} node[label] {Edwin I. Hatch};
\draw (-88.835, 40.17222222) node[mark2] {} node[label] {Clinton};
\draw (-93.23138889, 35.31027778) node[mark2] {} node[label, anchor=south west] {ANO};
\draw (-77.31, 43.27777778) node[mark2] {} node[label] {Ginna};
\draw (-91.78, 38.76166667) node[mark2] {} node[label, anchor=north west] {Callaway};
\draw (-86.56583333, 41.97527778) node[mark2] {} node[label, anchor=south east   ] {Donald C. Cook};

\end{tikzpicture}

\caption{Selected $k = 20$ of existing nuclear power plants in US, at most one per state.\label{fig:nucplants}}

\end{figure}

In metric spaces this naturally translates to seeking a set with large distances between its elements.
Of particular interest is the widely studied \emph{max-min diversification} (MMD) objective function (also known as \emph{remote-edge} or \emph{$p$-dispersion}), which asks for a set $S$ of size $k$ that maximizes the smallest pairwise distance between distinct items in $S$ \cite{erkut1990discrete,wang1988study}.
This problem is \NP-complete, but admits a greedy $1/2$ approximation in metric spaces \cite{tamir1991obnoxious,ravi1994heuristic}.


Recent work has introduced \emph{group fairness} constraints to the MMD problem \cite{Moumoulidou2020DiverseDS, Addanki0MM22, wang2022streaming, wang2023fair, Yanhao1, kurkure2024faster}. 
One important motivation is that in some data subset selection tasks, directly combining numerical and categorical attributes into a single objective function can be troublesome (although methods exist), as it might require
ad hoc decisions for discretizing numerical attributes or other pre-processing \cite{Moumoulidou2020DiverseDS}.
Instead, formulating fairness as a constraint allows to achieve a specified level of representation of certain groups that are defined through categorical attributes, while diversity aims at achieving a high distance between numerical feature vectors.

The assumption is that the universe of items $V = V_1 \cup \ldots \cup V_m$ is partitioned into $m$ disjoint groups (also called \emph{colors} or \emph{labels}).
A \emph{fair} solution is a set $S \subseteq V$ of fixed size $k$, while requiring that $S$ contains an appropriate amount of items from each color.

Specifically, \citet{Yanhao1} introduced the \emph{fair max-min diversification} (\FMMD) problem.  This problem models fairness by requiring $ \ell_i \leq |S \cap V_i| \leq u_i$ for each $i \in [m]$, where $ \ell_i$ and $u_i$ are specified parameters.
These constraints capture and generalize various existing notions of group fairness such as proportional representation \cite{celis2018fair,el2020fairness} and equal representation \cite{kleindessner2019fair}.
Aside from diversity maximization, these constraints have been used in various other settings like ranking \cite{celis2017ranking}, clustering \cite{kleindessner2019fair}, multiwinner voting \cite{celis2018fair} and submodular maximization \cite{el2020fairness}.

The \FMMD problem is formally described as:

\begin{prob}[Fair Max-Min Diversification (\FMMD) \cite{Yanhao1}]
\label{prob:fmmd}
Assume a metric space $(V,d)$ with distance function $d: V \times V \to \mathbb{R}_{\geq 0}$, a partition $V = V_1 \cup \ldots \cup V_m$ into $m$ colors, and parameters $k$, $ \ell_i$ and $u_i$ for each $i \in [m]$. Select a set $S \subseteq V$ of fixed size $|S| = k$, such that for all $i \in [m]:  \ell_i \leq |S \cap V_i| \leq u_i$, while maximizing
\begin{align}
\label{def:divscore}
\diver{S} = \min \limits_{\substack{u,v \in S \\ u \neq v}} d(u,v).
\end{align} 
\end{prob}

Here, $\diver{S}$ denotes the \emph{diversity score} of the set $S$.

\ptitle{FMMD with Few Representatives}
This paper studies \FMMD in case one is interested in only a \emph{few representatives} per color.

There are plenty of applications where this assumption is well-motivated. Indeed, consider the case where one requires \emph{at most one representative} from each color. Such fairness has been considered in prior work \cite{abbassi2013diversity,haxell2011forming,LOH2007904}. Specifically, the seminal paper of \citet{abbassi2013diversity} considers document presentation in the context of product search or news aggregator websites. They formulated a problem where no more than one document from the
same category (e.g., from the same news channel or product brand) may appear in the output.
Note that also the classic MMD problem (without fairness) is a special case of \FMMD with $k \leq m$, by assigning each point $v \in V$ with a unique color.

A second example is following facility location problem.
Suppose the U.S. government has decided to reduce its current nuclear power plants to $k$ plants.
They have a set of current locations from which they wish to select $k$ suitable options.
Each state has agreed to the placement of \emph{at most} one such plant within their borders. 
For safety reasons it is desirable that the pairwise distance between any two chosen locations is as large as possible. Where should they place the $k$ nuclear power plants? Figure~\ref{fig:nucplants} shows the output of our algorithm on the current 54 plants in 28 states to $k = 20$.

A final application is the formation of committees \cite{haxell2011forming}. For example, a university faculty wishes to form a committee, consisting of one member from each department in the faculty.
While this may seem like an easy task, there might be faculty members who are in severe conflict with each other.
Preferably they should not be on the committee together, in order to improve the decision making efficiency of the board. How to form a committee where no two members are in conflict?

\ptitle{Results} Our main contribution is a novel algorithm for \FMMD called $\algmmd{}$.
We show the following guarantee when $k \leq m$:

\begin{theorem}
\label{thm:mainres}
If $k \leq m$, then Algorithm~\ref{algo:final} approximates \FMMD within a factor of $\sqrt{\log(m)}/(3m)$ in polynomial time (see Section~\ref{sec:pet}).
\end{theorem}

We should point out that even though we focus on the case $k \leq m$, our algorithm can be immediately extended to the $k > m$ regime by adding $k - m$ artificial colors with $\ell_i = u_i = 0$. This comes at a cost of reducing the guarantee to $\sqrt{\log(\max(k, m))}/(3\max(k, m))$.

Table~\ref{table:summaryresults} shows a summary of the known approximation ratios for the \FMMD problem.
The first three algorithms listed have exponential running time. Among polynomial time algorithms, our algorithm improves upon the state-of-the-art $\frac{1}{m+1}$ approximation ratio from \citet{Addanki0MM22} by an additional $\bigO(\sqrt{\log m})$ factor.

We should note that the existing polynomial-time algorithms~\citep{Moumoulidou2020DiverseDS,Addanki0MM22} solve a variant of \FMMD with \emph{exact size group fairness constraints}.
These require exactly $k_i$ items from each color $i$, and thus a special case of Problem~\ref{prob:fmmd} with $ \ell_i = u_i$ for each color $i$, whereas we propose a polynomial-time algorithm that works with the general case.

Both algorithms~\citep{Addanki0MM22,Moumoulidou2020DiverseDS} use a similar approach that first builds well-separated clusters of nearby points, followed by an assignment phase to find at most one point from each cluster such that the fairness constraints are satisfied. However, their assignment phase only works for exact size group constraints.

It turns out that those algorithms can also be used for \FMMD (with same approximation guarantees), by replacing the assignment phase with a more intricate flow procedure described in Section~\ref{sec:assign}, which handles the \FMMD constraints.

\ptitle{Techniques} 
Our algorithm $\algmmd{}$ uses a similar high-level approach as \cite{Addanki0MM22} and \cite{Moumoulidou2020DiverseDS}. The idea is to construct well-separated clusters that are far apart, from which we select at most one point per cluster.
\citet{Moumoulidou2020DiverseDS} construct a graph where each connected component corresponds to such a cluster. Our approach also constructs a graph, but we further decompose each connected component by using Calinescu-Karloff-Rabani (CKR) random partitions \cite{calinescu2005approximation}. From the obtained decomposition, we remove all \emph{guard} vertices (certain boundary points) which ensures the final clusters are disconnected and thus diverse. We bound the probability that a fair solution is present in the obtained decomposition.
Our flow procedure from Section~\ref{sec:assign} will find such a feasible solution, if it exists. By repeating the algorithm enough times, we find a solution with good probability.


\begin{table}[t]
\setlength{\tabcolsep}{4pt}
\caption{Summary of the known approximation ratios for \FMMD.
A ratio of $\alpha \leq 1$ indicates the approximate solution has value at least $\alpha \times \OPT$.
Here $\rho$ is the ratio of the largest and the smallest distance, $m' = \max(m, k)$, and $T$ is the number of repetitions, affecting the success probability.
The asterisk indicates the results that only hold for the regime $k \leq m$.}
\label{table:summaryresults}
\centering
\begin{tabular}{lll}
\toprule
Paper & Approx. ratio & Run time \\
\midrule
\citep{Yanhao1} & $1$  & $\bigO(n^k \log n)$  \\
\citep{Moumoulidou2020DiverseDS} & $1/5$  & $\bigO(kn + k^2(em)^k)$ \\
\citep{Yanhao1} & $(1-\epsilon)/5$  & $\bigO(mkn+m^k \log \frac{1}{\epsilon})$\\
\citep{Moumoulidou2020DiverseDS} & $1/(3m-1)$  & $\bigO(kn + k^2m^2\log (km))$ \\
\citep{Addanki0MM22} & $1/((m+1)(1+\epsilon))$  &  $\bigO(nkm^3 \epsilon^{-1}\log n)$\\
\textbf{ours}* & $\sqrt{\log m}/(3m(1 + \epsilon))$  &  $\bigO(Tm^2n^2 \epsilon^{-1} \log \rho)$\\
\textbf{ours}* & $\sqrt{\log m}/(5m(1 + \epsilon))$  &  $\bigO((kmn + Tm^3k^2) \epsilon^{-1} \log \rho)$\\
\textbf{ours} & $\sqrt{\log m'}/(3m'(1 + \epsilon))$  &  $\bigO(Tm'mn^2 \epsilon^{-1} \log \rho)$\\
\textbf{ours} & $\sqrt{\log m'}/(5m'(1 + \epsilon))$  &  $\bigO((kmn + Tm'm^2k^2) \epsilon^{-1} \log \rho)$\\
\bottomrule
\end{tabular}
\end{table}

\section{Related Work}

\ptitle{MMD and MSD} The most fundamental problem settings for diversity maximization assume the input is given by a metric space $(V,d)$ with nonnegative distance function $d: V \times V \to \mathbb{R}_{\geq 0}$ and an integer $k$.
The goal is to select a set $S \subseteq V$ of fixed size $|S|=k$ which maximizes a certain objective function.

The two most popular objective functions are \emph{max-sum diversification} (MSD) and \emph{max-min diversification} (MMD).
MSD maximizes the sum of the pairwise distances in $S$ (also known as \emph{remote-clique}), while MMD maximizes the minimum pairwise distances in $S$ (also known as \emph{remote-edge} or \emph{$p$-dispersion}).
It has been argued that the MMD objective is more suited for tasks like summarization and feature selection, because MSD has a tendency to select outliers and include highly similar items, whereas MMD covers the data range more uniformly \cite{Yanhao1, Moumoulidou2020DiverseDS}.
In this work, we focus on MMD.

The classic MMD problem asks for a set of $k$ items under the max-min objective, without any fairness constraints. 
Both MMD and MSD are shown to be \NP-complete by \cite{erkut1990discrete} and \cite{wang1988study}.
\citet{tamir1991obnoxious} and \citet{ravi1994heuristic} showed that the greedy ``next-furthest-point-strategy'' yields a $\frac{1}{2}$-approximation for MMD, which is the well-known GMM algorithm.
The GMM algorithm is identical to the Gonzalez-heuristic \cite{gonzalez1985clustering} which gives a 2-approximation for the $k$-center problem. \citet{ravi1994heuristic} showed that for any $\epsilon>0$ MMD is \NP -hard to approximate within a factor of $\frac{1}{2}+\epsilon$, thereby showing the optimality of the Gonzalez-heuristic in terms of worst-case performance.
Additionally, \citet{ravi1994heuristic} showed that MMD without the triangle inequality constraints is \NP -hard to approximate within any multiplicative factor. \citet{amagata2023diversity} consider MMD in the presence of outliers. \citet{kumpulainen2024max} have studied the MMD problem with asymmetric distances between the points, while still retaining the directed triangle inequality. They provided a $\frac{1}{6k}$ approximation for this asymmetric generalization.
MMD has also been studied in the dynamic setting \cite{amagatadyn, drosou2013diverse}.

\ptitle{Fair MSD} Fairness for the MSD objective is well-understood and admits a $1/2-\epsilon$ guarantee via local search algorithms in polynomial time \cite{abbassi2013diversity,borodin2017max}.
This is in contrast with the fair MMD problem formulations, which appear to be more difficult to approximate.
Attaining a ratio larger than $1/2$ for fair MSD is hard under the planted clique hypothesis \cite{nipslinear}.
Most authors provide results for general (partition) matroids constraints \cite{nipslinear,borodin2017max,abbassi2013diversity, Ceccarello}, which includes the fairness case when exactly $k_i$ items are required from each color \cite{Moumoulidou2020DiverseDS, Addanki0MM22, wang2022streaming, wang2023fair}.
\citet{cevallos2017local} extend the setting to distances of negative type and also use local search to design a  $\bigO(1-1/k)$ approximation.

\ptitle{Fair MMD} The first to incorporate fairness constraints into the MMD problem was \citet{Moumoulidou2020DiverseDS}.
They study a variant of \FMMD which considers exact size group fairness constraints (see Section~\ref{sec:intro}).
\citet{Moumoulidou2020DiverseDS} provided a $\frac{1}{4}$ approximation for the case of $m=2$ colors, a $\frac{1}{5}$ approximation for the case where $m=\bigO(1)$ and $\sum_i k_i=o(\log n)$, and a $\frac{1}{3m-1}$ approximation for the general case.
This was subsequently improved to the current state-of-the-art $\frac{1}{m+1}$ approximation guarantee by \citet{Addanki0MM22}.
The authors also showed a $\frac{1}{2}$ approximation using an LP-based randomized rounding method, but which only satisfies the fairness constraints in expectation.
Using concentration bounds---which requires that each $k_i$ is large enough---the same randomized rounding method can be used to obtain a $\frac{1}{6}$ approximation while guaranteeing that $(1-\epsilon)k_i$ points are selected from each color.

Specifically for the case for Euclidean metrics, \citet{Addanki0MM22} proved that the problem can be solved exactly in one dimension.
For constant dimensions, they give a $1+\epsilon$ approximation that runs in time $\bigO(nk)+2^{\bigO(k)})$ where $k = \sum_i k_i$.
They showed how to reduce the running time to a polynomial in $k$ at the expense of only picking $(1-\epsilon)k_i$ points from each color.
The running time was recently improved by \citet{kurkure2024faster} to linear in $n$ and $k$ using the multiplicative weight update on an associated linear program in combination with geometric data structures.
The exact-size \FMMD variant has also been studied in the context of streaming and sliding window models \cite{wang2022streaming,  Addanki0MM22, wang2023fair}.

The (general) \FMMD problem (Problem~\ref{prob:fmmd}) was introduced by \citet{Yanhao1}.
They propose an exact algorithm with runtime $\bigO(n^k \log n)$, by reducing it to the problem of finding a maximum independent set in an undirected graph (subject to fairness constraints), which is solved by an integer linear program. 
Additionally, using coreset constructions they obtain a $\frac{1-\epsilon}{5}$-approximation that runs in exponential time. They empirically show their algorithm scales to large datasets when $m$ and $k$ are not too large.

We should note that
\citet{Moumoulidou2020DiverseDS} also argued that their algorithm extends to (partition) matroid constraints. However, since the \FMMD fairness constraints do not have matroid structure, the algorithms by \citet{Addanki0MM22} and \citet{Moumoulidou2020DiverseDS} are not directly applicable to the \FMMD problem.

Interestingly, the best lower bounds for approximating these fair MMD problem formulations are still identical to bounds for unconstrained MMD. There remains a large gap between the polynomial time approximation ratios from Table~\ref{table:summaryresults} and the $\frac{1}{2}+\epsilon$ lower bound.



\ptitle{Independent Transversals} Theorem~\ref{thm:mainres} gives an approximation guarantee for \FMMD when $k \leq m$.
This assumption includes cases where one requires at most a constant number of items from each color.
Such fairness is closely related to finding an \emph{independent transversal} (IT) in a graph.\!\footnote{Given a graph and a partition $V_1 \cup \ldots \cup V_m$ of its vertices, an IT is an independent set of vertices $v_1,\ldots,v_m$ such that $v_i \in V_i$ for each $i$.}
Finding an IT is known to be \NP -complete \cite{haxell2011forming, alon1988linear,fellows1990transversals}.
The IT problem has many natural applications where one is interested in selecting a unique representative from each group (the \emph{transversal} requirement), while asking that the representatives are in some sense far apart (the \emph{independence} requirement) \cite{haxell2011forming}.
Our paper studies the optimization variant of the IT problem in metric spaces with a max-min objective.

\section{Algorithm}

Our algorithm consists of 3 phases:
$(i)$ Given a universe $V$ we select a subset $U$ of $V$ such that $U$ still has an approximate solution and any two points in $U$ of the same color are far away from each other.
$(ii)$ We randomly partition $U$ in clusters such that the points of a feasible solution, with a certain probability, will be in their own clusters and the points from different clusters are far away.
$(iii)$ We use maximum flow to assign clusters to colors, and extract points of those colors from the clusters. It follows that, if the second step was successful, the resulting set will satisfy the constraints.

We present two variants of our algorithm: a slower version with a better guarantee and a faster version with a worse guarantee. The difference between the two algorithms is the preprocessing step.
We assume that we know the optimal score $\OPT$.
In Section~\ref{sec:pet} we use grid search to find an (approximate) guess for $\OPT$.

\subsection{Preprocessing step by pruning
\label{sec:prep}}

We preprocess $(V,d)$ by a simple clustering step described in Algorithm~\ref{algo:preprocessing}. The algorithm has two parameters, a distance threshold $\gamma$ and a budget $b$.

$\algpre{V, d, \gamma, b}$ iterates over the colors. For each color $i$, the algorithm iteratively selects points and removes all points that are within distance of $\gamma$ of the selected ones. This is repeated until no points remain or we have selected $b$ points of that color.

\begin{algorithm}[t]
\caption{$\algpre{V, d, \gamma, b}$, finds a subset of $V$ ensuring same color points are far apart while still having an approximate solution}\label{algo:preprocessing} 
\begin{algorithmic}[1]
\Require Space $(V_1 \cup \ldots \cup V_m, d)$, parameter $\gamma>0$, and a budget $b$.
\For{each color $i$}
 \State Label all $u \in V_i$ as unmarked.
 \State Set $U_i \leftarrow \emptyset$.
\While {there is unmarked point in $V_i$ and $\abs{U_i} \leq b$}
	\State Pick any such unmarked $v$.
	\State Mark all vertices in $\{u \in V_i: d(u,v) < \gamma \}.$ 
	\State $U_i \leftarrow U_i + v$.
\EndWhile
\Ensure $U = U_1 \cup \cdots \cup U_m$.
\EndFor

\end{algorithmic}

\end{algorithm}

An immediate consequence of the algorithm is that the points of the same color in $U$ will be far away from each other.

\begin{lemma}
Assume $(V, d)$, a parameter $\gamma$, and a budget $b$.
Let $U = \algpre{V, d, \gamma, b}$. Then any points $u, v \in U$ with the same color have $d(u, v) \geq \gamma$.  
\end{lemma}

This reduction comes at a cost of the optimal solution. More specifically, if we do not set budget $b$ and set $\gamma = \OPT / 3$, then the optimum value in the resulting set may reduce by a factor of three.

\begin{proposition}
\label{prop:pre1}
Assume $(V, d)$ with $m$ colors and $n$ points, $V = V_1 \cup \cdots \cup V_m$. Assume $\tau > 0$ such that there is
$O \subseteq V$ satisfying the constraints with $\diver{O} \geq \tau$. Let $U = \algpre{V, d, \tau / 3, n}$. 
Then there is $S \subseteq U$ satisfying the constraints with $\diver{S} \geq \tau / 3$.
\end{proposition}

\begin{proof}
For each $o \in O$, let $f(o)$ be the closest point in $U$ with the same color. Note that $f$ is injective as otherwise we violate the triangle inequality. Consequently, the set $S = \set{f(o) \mid o \in O}$ satisfies the constraints. Moreover, the triangle inequality implies that $\diver{S} \geq \tau / 3$, proving the claim.
\end{proof}

Assuming that we can compute the distance $d(u, v)$ in unit time, the running time of $\algpre{V, d, \tau / 3, n}$ is in $\bigO(n^2)$, which may be too expensive for large $n$.

If this is an issue we can speed up the preprocessing by only considering $k$ points from each color. This comes with additional cost of reducing the approximation by a factor of five.

\begin{proposition}
\label{prop:pre2}
Assume $(V, d)$ with $m$ colors, $V = V_1 \cup \cdots \cup V_m$. Assume $\tau > 0$ such that there is $O \subseteq V$ satisfying the constraints with $\diver{O} \geq \tau$. Let $U = \algpre{V, d, 2 \tau / 5, k}$. 
Then there is $S \subseteq U$ satisfying the constraints with $\diver{S} \geq \tau / 5$.
\end{proposition}

The proof is an adaptation of the proof for Theorem~6 in~\citep{Moumoulidou2020DiverseDS}.

\begin{proof}
We say that color $i$ is critical if $\abs{U_i} < k$.
For $o \in O$ with critical color define $f(o)$ to be the closest point to $o$ in $U$ with the same color. Let 
\[
    S = \set{f(o) \mid o \in O, o \text{ has critical color}}\quad.
\]
The triangle inequality now implies $\diver{S} \geq \tau / 5$.

Let $c_i$ be the number of points of color $i$ in $O$.

Process non-critical colors as follows. Let $j$ be a non-critical color. Delete points from $U_j$ that are within $\tau / 5$ from any point in $S$. Since any two points in $U_j$ are least $2 \tau / 5$ far away, we delete at most $\abs{S} \leq k - c_j$ points. Therefore, there are at least $c_j$ points left in $U_j$. Add $c_j$ points to $S$.
Performing this step for each non-critical color (in any order), will result in $S$ that has the same color counts than $O$, and $\diver{S} \geq \tau / 5$.
\end{proof}

Assuming that we can compute the distance $d(u, v)$ in unit time, the running time of $\algpre{V, d, 2\tau / 5, k}$ is in $\bigO(n km)$. The resulting set contains $\bigO(km)$ points.

\subsection{Random decomposition of the pruned space\label{sec:decomp}}

Our next goal is to decompose $U$ into clusters that are far enough from each other and that a feasible solution $O$ will have at most one point in each cluster.

Assume that we have guessed the score of the optimal solution in $U$, say $\gamma$. Now let $\alpha = \sqrt{\log(m)} / m$, and define a graph $G = (U, E)$, consisting of the points and edges $(u, v)$ for which $d(u, v) < \gamma \alpha$.

Denote $d_G(u, v)$ to be the shortest path distance in $G$ between two points $u, v$. Let
\[
    B_G(u, R) = \set{v \in U \mid d_G(u, v) \leq R}
\]
be the points in $G$ with shortest path distance of at most $R$ from $u$.

Let $\funcdef{\pi}{[n]}{U}$ be a random permutation of $U$. Let us define $\Delta_1 = \max(\floor{1/(4\alpha)}, 1)$ and $\Delta_2 = \floor{1/(2\alpha)}$. It is easy to verify that $\Delta_1 \leq \Delta_2$ for $m \geq 1$. Select $R$ uniformly between $\Delta_1$ and $\Delta_2$.

To define the decomposition we first partition $U$ into clusters $C_j$ for $j \in [n]$. Each $C_j \subseteq B_G(\pi(j), R)$ consists of those points in $B_G(\pi(j), R)$ that have not been included in any $B_G(\pi(i), R)$, where $i < j$. Note that some $C_j$ might be empty.

This step is known as a Calinescu-Karloff-Rabani (CKR) random partition of a metric space \cite{calinescu2005approximation}, with a slight parameter modification to fit our needs.
CKR partitions play a pivotal role in metric embedding theory \cite{FAKCHAROENPHOL2004485,MendelFOCS}.

Our final decomposition is $D_1, \ldots, D_n$, where $D_j$ contains vertices in $C_j$ that have $d_G(\pi(j), u) < R$. The vertices in $C_j \setminus D_j$ are refered to as \emph{guard} vertices and are removed from the graph.

The algorithm is given in Algorithm~\ref{algo:padded}.

The following proposition states that with a certain probability a fair solution can be found in clusters $D_1, \ldots, D_n$.

\begin{algorithm}[t]
\caption{$\algdec{U, d, \gamma}$, padded decomposition of space $U$.}\label{algo:padded} 
\begin{algorithmic}[1]
\Require Space $(U_1 \cup \cdots \cup U_m, d)$ with $n$ points, parameter $\gamma$.
\State $\alpha \define \sqrt{\log m} / m$.
\State $G \define (U,E)$ with $E = \set{ (u, v) \in U \times U \mid d(u, v) < \gamma \alpha}$.
\State $\Delta_1 \define \max(\floor{1/(4\alpha)}, 1)$.
\State $\Delta_2 \define \floor{1/(2\alpha)}$.

\State $\pi \define$ a uniformly random permutation of $U$.
\State $R \define$ a uniformly random value in $\enset{\Delta_1, \Delta_1 + 1}{\Delta_2}$.
\State $Z \define \emptyset$
\For {$j = 1, \ldots, n$}
\State $C_j \define B_G(\pi(j),R) \setminus Z$.
\State $Z \define Z\cup C_j$.
\State $D_j \define \set{u \in C_j \mid d_G(\pi(j), u) < R}$.
\EndFor

\Ensure $\{D_1,\ldots,D_n\}$. 
\end{algorithmic}
\end{algorithm}

\begin{proposition}
\label{prop:decompose}
Assume $(U, d)$ with $m$ colors, $U = U_1 \cup \cdots \cup U_m$. Assume $\gamma > 0$ such that there is $O \subseteq U$ with $\diver{O} \geq \gamma$. Moreover, assume that for any $u, v \in U_i$, it holds that $d(u, v) \geq \gamma$. If $k \leq m$, then with probability of $\Omega(1/m)$, 
$\algdec{U, d, \gamma}$ yields $\enset{D_1}{D_n}$
such that $\abs{O \cap D_i} \leq 1$ for every $i$ and $O \subseteq \bigcup{D_i}$.

\end{proposition} 

\begin{proof}
Adopt the notation used in Algorithm~\ref{algo:padded}.
Let $u, v \in C_j$. Note by definition, the shortest path from $u$ to $v$ in $G$ is at most $2 \Delta_2$. Due to the triangle inequality, $d(u, v) < \gamma$. This immediately implies that $u$ and $v$ must have different colors, and that $\abs{C_j \cap O} \leq 1$.

Next we bound the probability that $O$ does not have any guards.

First note that since $\alpha \leq 1/2$ for any $m$, we have $\Delta_2 \geq 1$, which guarantees $\Delta_1 \leq \Delta_2$.
Let $o \in O$. A known result, see, for example, Lemma~1.2 by~\citet{har2018two}, states that
\[
    p(o \text{ is a guard}) \leq \frac{1}{\Delta_2 - \Delta_1 + 1} \log \frac{\abs{B_G(o, \Delta_2)}}{\abs{B_G(o, \Delta_1 - 1)}}\quad.
\]

Note that
\[
    \Delta_2 - \Delta_1 + 1 = \floor{1/(2\alpha)} - \max(\floor{1/(4\alpha)}, 1) + 1
    \geq 1 / (4 \alpha).
\]

Let $u, v \in B_G(o, \Delta_2)$. The shortest path from $u$ to $v$ in $G$ is at most $2 \Delta_2$. Due to the triangle inequality, $d(u, v) < \gamma$. Therefore, $u$ and $v$ must have different colors, and consequently, $\abs{B_G(o, \Delta_2)} \leq m$.

If $o$ is in a connected component, say $W$, of $G$ with less than $\Delta_1$ vertices, then $W$ has no guards, and therefore $o$ cannot be a guard. Assume $\abs{W} \geq \Delta_1$. Then $\abs{B_G(o, \Delta_1 - 1)} \geq \Delta_1 \geq 1/ (4 \alpha)$.

In summary, if we
let $\omega = 4 \alpha \log (4m \alpha)$, then
\[
    p(o \text{ is a guard}) \leq \omega\quad.
\]

Note that $o$ being a guard depends solely on the order of nodes in $B_G(o, \Delta_2)$. Let $o, o' \in O$ be two distinct nodes. Then since $B_G(o, \Delta_2)$ and $B_G(o', \Delta_2)$ are disjoint, the event of $o$ being a guard is independent of $o'$ being a guard. Thus
\[
    p(O \text{ not having guards}) = \prod_{o \in O} p(o \text{ is not a guard}) \leq (1 - \omega)^{k}\quad.
\]

There is a constant $M$ such that for any $m \geq M$, it holds that
$\omega < 1$, $1 - \omega < 1/2$, and $\log m \geq 2 \omega m$. For such $m$, we have
\[
\begin{split}
    (1 - \omega)^k & \geq (1 - \omega)^m
    \geq \exp\pr{\frac{-\omega m}{1 - \omega}} \\
    & \geq \exp\pr{-2\omega m} = m^{-2 \omega m / \log m} \geq 1/m, \\
\end{split}
\]
where the first inequality uses the assumption $k \leq m$. 
\end{proof}

Next we show that the clusters $D_1, \ldots, D_n$ are far away from each other.

\begin{proposition}
\label{prop:diverse}
Let $\enset{D_1}{D_n}$ be given by $\algdec{U, d, \gamma}$. Let $S$ be any set for which $\abs{S \cap D_i} \leq 1$ for every $D_i$, then $\diver{S} \geq \gamma \alpha$.
\end{proposition}

\begin{proof}
Adopt the notation used in Algorithm~\ref{algo:padded}.
Assume a cross-edge $(u, v)$ in $G$ with $u \in C_i$ and $v \in C_j$ where $i < j$. Then $u$ must be a guard as otherwise $v \in C_i$.
Thus, the clusters without the guards are disconnected. Assume now that $u \in D_i$ and $v \in D_j$ where $i < j$, then $d(u, v) \geq \gamma \alpha$.
\end{proof}

Let us now look at the computational complexity of $\algdec{}$. The bottleneck of the algorithm is computing $B_G(\pi(j), R)$. 
Computing $B_G$ independently leads to a running time of $\bigO(n^3)$. We can speed-up computing $B_G(\pi(j), R)$
by ignoring the nodes $v$ that have shorter distance $d_G(\pi(i), v)$ than the distance $d_G(\pi(j), v)$ for some $i < j$. Then each node is only visited $\bigO(R)$ times, leading to computational complexity $\bigO(Rn^2) \subseteq \bigO(m n^2)$.
In practice, a naive approach is reasonable.

\subsection{Assigning colors to clusters\label{sec:assign}}

We will now describe our final step. In order to do this, assume that $\algdec{}$ has produced a decomposition $\enset{D_1}{D_n}$ such that there is a solution $S \subseteq \bigcup D_i$ satisfying all the constraints and that every node in $S$ is in its own cluster $D_i$, that is, $\abs{D_i \cap S} \leq 1$.

Our final step is to extract $S$, or any solution satisfying the constraints, from $\enset{D_1}{D_n}$.
We will do this by assigning (some of the) clusters $\set{D_i}$ to colors. The assignment will done such that each cluster has a vertex of assigned color, and that there are enough assignments to satisfy all the constraints.

In order to do this, we define a flow network $H = (W, A)$.
An illustration of $H$ is given in Figure~\ref{fig:flow}.
The vertices $W$ contains $n$ nodes, each corresponding to the cluster $D_i$. Moreover, $W$ contains $m$ nodes $c_1, \ldots, c_m$, corresponding to the colors. Finally, the $W$ has additional three nodes $s$, $z$, and $t$.

The edges $A$ are as follows. We connect $s$ to each $D_i$ with unit capacity. For each vertex $u$ in $D_i$ with color $j$, we add an edge from $D_i$ to $c_j$ with unit capacity. Each $c_j$ is connected to $t$ with capacity $\ell_i$ and to $z$ with capacity $u_j - \ell_j$. Finally, $z$ is connected to $t$ with capacity $k - \sum \ell_i$.

\begin{figure}
\begin{tikzpicture}[yscale=0.9]

\tikzstyle{edge}=[->, >=latex,  semithick, yafcolor5, in=180, out=0]

\tikzstyle{node}=[inner sep=1pt]
\tikzstyle{lab}=[inner sep=1pt, above, sloped, black]
\node[node] (s) at (0, 0.5) {$s$};
\node[node] (d1) at (2, 2) {$D_1$};
\node[node] (d2) at (2, 1) {$D_2$};
\node[node] (d3) at (2, -0.5) {$D_n$};
\node[node]  at (2, 0.375) {$\vdots$};

\node[node] (c1) at (4, 2) {$c_1$};
\node[node] (c2) at (4, 1) {$c_2$};
\node[node] (c3) at (4, -0.5) {$c_m$};
\node[node]  at (4, 0.375) {$\vdots$};

\node[node] (slack) at (7, 1.5) {$z$};
\node[node] (t) at (7, 0) {$t$};

\draw[edge] (s) edge (d1);
\draw[edge] (s) edge (d2);
\draw[edge] (s) edge (d3);

\draw[edge] (d1) edge (c1);
\draw[edge] (d1) edge (c2);

\draw[edge] (d2) edge (c3);

\draw[edge] (d3) edge (c2);

\draw[edge] (c1) edge node[lab, pos=0.4] {$u_1 - \ell_1$} (slack);
\draw[edge] (c2) edge node[lab, pos=0.17] {$u_2 - \ell_2$} (slack);
\draw[edge] (c3) edge node[lab, pos=0.3] {$u_m - \ell_m$} (slack);

\draw[edge] (c1) edge node[lab, pos=0.35] {$\ell_1$} (t);
\draw[edge] (c2) edge node[lab, pos=0.3, below] {$\ell_2$} (t);
\draw[edge] (c3) edge node[lab, pos=0.5] {$\ell_m$} (t);

\draw[edge, bend left] (slack) edge node[pos=0.5, inner sep=1pt, anchor=west, black] {$k - \sum \ell_i$}  (t);

\end{tikzpicture}
\caption{Flow network used to solve the assignment of clusters $\set{D_i}$ to colors. Capacities are indicated at edges. Unmarked edges have capacity of 1.}
\label{fig:flow}
\end{figure}
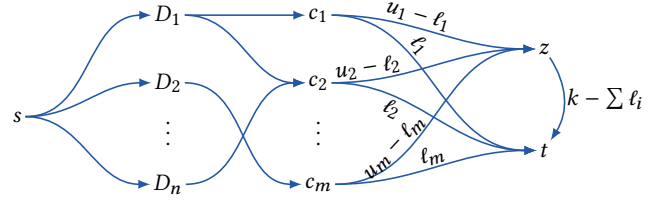

\begin{algorithm}[t]
\caption{$\algflow{\set{D_i}, \set{\ell_i}, \set{u_i},  k}$ tries to find a feasible solution from $\set{D_i}$}\label{algo:flow} 
\begin{algorithmic}[1]
\Require Clusters $D_1, \ldots, D_n$,
constraints $\set{\ell_i}$, $\set{u_i}$, and $k$.

\State construct flow network $H = (W, A)$.
\State solve the (integral) maximum flow problem for $H$.
\State $S \define \emptyset$
\For {$i = 1, \ldots, n$}
    \If {$D_i$ flows to $c_j$}
        \State $u \define $ any vertex in $D_i$ with color $j$.
        \State add $u$ to $S$.
    \EndIf
\EndFor

\Ensure $S$
\end{algorithmic}
\end{algorithm}

The following proposition states that we can extract the solution.

\begin{proposition}
\label{prop:assign}
Assume that there is a feasible set $O$ such that
$\abs{D_i \cap O} \leq 1$ and $O \subseteq \bigcup{D_i}$. Then $\algflow{\{D_i\}, \set{\ell_i}, \set{u_i}, k}$ yields a set $S$ satisfying the constraints.
\end{proposition}

\begin{proof}
Consider the following flow. Assume $D_i \cap O = \set{u}$, where $u$ has a color $j$. Then push flow from $s$ to $c_j$ via $D_i$. Let $s_j$ be the number of vertices in $O$ with color $j$. Push a flow of $\ell_j$ from $c_j$ to $t$ and a flow of $s_j - \ell_j$ from $c_j$ to $z$. Finally, push a flow of $z$ from $t$, at maximum capacity. Since $O$ satisfies the constraints, we see that the flow is valid. Moreover, all edges pointed to $t$ are saturated, making this flow optimal with a throughput of $k$.

Let $f$ be an optimal flow. Since the capacities are all integral, we can assume that the flow is integral. This flow must have a throughput of $k$, which can only happen if the edges to $t$ are saturated. Consequently, there
are at least $\ell_i$, and at most $u_i$, clusters flowing to $c_i$. Moreover, there are in total $k$ clusters flowing. Consequently, $S$ constructed by $\algflow{}$ satisfies the constraints.
\end{proof}

Let us consider the computational complexity of $\algflow{}$.
Since every point has only one color, the graph $H$ has $\bigO(n + m)$ edges. A deterministic algorithm by~\citet{van2023deterministic} can solve the flow in $\bigO((n + m)^{1 + o(1)} \log^2 k)$ time. In practice, a standard flow solver such as push-relabel, suffices.

\subsection{Putting everything together}
\label{sec:pet}

Our next step is to combine all three steps together with the main algorithm given by Algorithm~\ref{algo:final}.
The algorithm first preprocesses the space using $\gamma_1$ as the threshold, decomposes the space using $\gamma_2$, and tries to find a feasible solution using maximum flow. We still need to determine $\gamma_1$ and $\gamma_2$, which we will discuss shortly.

\begin{algorithm}[t]
\caption{$\algmmd{V, d, \gamma_1, \gamma_2, b, T}$, main algorithm combining all three steps}\label{algo:final} 
\begin{algorithmic}[1]
\Require Space $(V = V_1 \cup \ldots \cup V_m, d)$, parameters $\gamma_1$, $\gamma_2$, preprocessing budget $b$, constraints $\set{\ell_i}$, $\set{u_i}$, $k$, maximum repetitions $T$.
\State $(U,d) \define \algpre{V, d, \gamma_1, b}$.
\State $S \define \emptyset$.
\MRepeat{ $T \times m$ times or until $S$ satisfies the constraints}
\State $\enset{D_1}{D_n} \define \algdec{U, d, \gamma_2}$.
\State $S \define \algflow{\set{D_i}, \set{\ell_i}, \set{k_i},  k}$.
\EndRepeat

\Ensure $S$
\end{algorithmic}
\end{algorithm}

Next we state the approximation ratio and the computational complexity of $\algmmd{}$. The first set of parameters lead to a slower algorithm due to $\bigO(n^2)$ term but yields a better guarantee.

\begin{proposition}
\label{prop:together1}
Assume a space $(V, d)$ with $n$ points with a set $O$ satisfying the constraints and having a diversity score $\diver{O} \geq \tau$.
Then $\algmmd{V, d, \tau / 3, \tau / 3, n, T}$ yields a solution $S$ with $\diver{S} \geq \sqrt{\log m}\tau / (3m)$
with probability $1 - 2^{-\Omega(T)}$. The computational complexity of the algorithm is in $\bigO(Tm^2n^2)$.
\end{proposition}

\begin{proof}
Propositions~\ref{prop:pre1},~\ref{prop:decompose},~\ref{prop:diverse},~and~\ref{prop:assign} imply the approximation guarantee. Repeating the decomposition phase (line 3 in Algorithm~\ref{algo:final}) together with Proposition~\ref{prop:decompose} gives an overall success probability of
\[
1-(1-\Omega(1/m))^{Tm} \geq 1 - 2^{-\Omega(T)}.
\]
The running time of the algorithm is equal to
\[
\bigO(n^2 + Tm^2n^2 + Tm(n + m)^{1 + o(1)} \log^2 k) \in \bigO(Tm^2n^2),
\]
proving the claim.
\end{proof}

The second set of parameters lead to a faster algorithm but yields a worse guarantee.

\begin{proposition}
\label{prop:together2}
Assume a space $(V, d)$ with $n$ points with a set $O$ satisfying the constraints and having a diversity score $\diver{O} \geq \tau$.
Then $\algmmd{V, d, 2\tau / 5, \tau / 5, k, T}$ yields a solution $S$ with $\diver{S} \geq \sqrt{\log m}\tau / (5m)$
with probability $1 - 2^{-\Omega(T)}$. The computational complexity of the algorithm is in $\bigO(kmn + T m^3 k^2)$.
\end{proposition}

\begin{proof}
Propositions~\ref{prop:pre2},~\ref{prop:decompose},~\ref{prop:diverse},~and~\ref{prop:assign} imply the approximation guarantee. Repeating the decomposition phase together with Proposition~\ref{prop:decompose} guarantees the probability. The running time of the algorithm is equal to 
\[
    \bigO(kmn + Tm(km)^2 + Tm(km + m)^{1 + o(1)} \log^2 k) \in \bigO(kmn + Tm^3 k^2),
\]
proving the claim.
\end{proof}

Both propositions require $\tau$, the optimal score. A naive option would be to test all existing distances, multiplying the running time by a factor of $\bigO(n^2)$. A faster alternative option is to do a grid search by considering candidates of form $\tau_{i} = \tau_0 (1 + \epsilon)^i$, where $\tau_0$ is a distance between any two points and $i$ is an integer. We can lower bound $i$ as we know that if $\tau_i$ is the smallest non-zero possible distance, then \algmmd{} will find a solution. On the other hand, if $\tau_i$ yields $\gamma_2$ that is larger than the largest distance in the pruned space, then there is no solution. Let $\rho$ be the ratio of the largest distance and the smallest non-zero distance in $V$. Then we need only $\bigO(\log_{1 + \epsilon} \rho) \subseteq \bigO(\epsilon^{-1} \log \rho)$ tests.
This comes at a cost to the quality of the solution by a factor of $1/(1 + \epsilon)$, proving the approximation results stated in Table~\ref{table:summaryresults}.

\ptitle{The $k > m$ regime} 
We can extend $\algmmd{}$ to work with $k > m$ by adding $k - m$ artificial colors with with $u_i = \ell_i = 0$. This comes yields a weaker guarantee $\bigO(\sqrt{\log(m')}/(3m'))$, where $m' = \max(m, k)$. Moreover, the running times of the two algorithms change to 
$\bigO(Tmn^2 m')$
and
$\bigO(kmn + Tm^2 k^2 m')$.

Alternatively, we can obtain the original guarantees, but with a slower algorithm. To this end, note that
Proposition~\ref{prop:decompose} proves a success probability of $\Omega(1/m)$ for a single repetition of line 3 in Algorithm~\ref{algo:final}, in case $k \leq m$. Assume that  $k$ is proportional to $m$, say $k \leq \beta m$ for some $\beta \geq 1$. Proposition~\ref{prop:decompose} would then yield a success probability of $\Omega(1/m^{\beta})$.
Therefore, Algorithm~\ref{algo:final} (line 3) would need $T \times m^{\beta}$ repetitions to ensure that Algorithm~\ref{algo:final} succeeds with probability $1 - 2^{-\Omega(T)}$ (see Propositions~\ref{prop:together1} and \ref{prop:together2}). The running time of such algorithm is polynomial only when $\beta$ remains a constant.

\subsection{Practical improvements}
\label{prac:improv}

In order to improve the quality of the solution in practice, we do the following additional steps.

During preprocessing we do not treat colors individually. Instead we always select the point of a color that we need that is furthest from already added points. This increases the preprocessing time by $\bigO(\log n)$ due to sort.

When forming $G$ in \algdec{} it turns out that the threshold $\gamma_2 \alpha$ is too small, resulting in subpar results. To combat this we vary $\gamma_2$ independently, instead of having it fixed to $\tau$. More specifically, we vary $\gamma_2$ between $\gamma_1 / 2$ and $\gamma_1 / \alpha$ using a similar $(1 + \epsilon)$ grid as with $\tau$.
This increases the running time of the last two steps by $\bigO(\epsilon^{-1} \log \alpha) \subseteq \bigO(\epsilon^{-1} \log m)$.

Finally, we found out that in practice \algmmd{} yields good results after only a few iterations. Thus in experiments, we perform \algdec{} only a constant number of times.
\section{Experimental evaluation}
In this section, we describe our experimental evaluation.
All experiments were performed on an Intel Core i5 machine at 1.6 GHz with 16 GB RAM. Our methods are implemented in Python 3.8 and will be made publicly available.

\ptitle{Data} For data, we use the four real-world public datasets that were previously used for evaluation by \citet{Yanhao1} in their work on \FMMD.
The dataset characteristics are listed in Table~\ref{tab:datasets}.
Here, \emph{dim} refers to the dimension of the numerical feature vectors. As a distance measure, we use euclidean distance. For the \emph{CelebA} dataset, the original dimension of the feature vectors was $25\,088$. We used sparse random projections \cite{achlioptas2001database} to reduce the dimension to 10.
In all datasets, labels are defined from demographic
attributes such as sex, race and age.
For more details on data preparation, and other statistics such as relative proportion of each label within the data, see \cite[Appendix A]{Yanhao1}.

\ptitle{Baselines} We compare our algorithm $\algmmd{}$ with the three most scalable methods listed in Table~\ref{table:summaryresults} (Section~\ref{sec:intro}).
The polynomial time approximation algorithm from \cite{Moumoulidou2020DiverseDS} is denoted as $\fairflow{}$. The approximation algorithm from \cite{Addanki0MM22} is denoted as $\greedyflow{}$.
The scalable coreset approximation based on integer linear programming from \cite{Yanhao1} is denoted as $\silp{}$.
We use the publicly available Python implementations from \cite{Yanhao1} for these methods. 
The other methods from Table~\ref{table:summaryresults} are exponential time algorithms that are not scalable to large datasets \cite{Yanhao1} and are not used for evaluation.

\ptitle{Parameters} We use similar parameter selection as in \cite{Yanhao1}. The fairness constraints are defined according to \emph{proportional representation} \cite{celis2018fair}. Namely, for each label $i \in [m]$ we set $\ell_i = \max(1,(1-\alpha)k\frac{|V_i|}{|V|})$ and $u_i = \max(1,(1+\alpha)k|V_i|/n)$ for $\alpha=0.2$.
For $\fairflow{}$ and $\greedyflow{}$ we set $k_i = \left \lceil{k|V_i|/n}\right \rceil$ or $k_i = \left \lfloor{k|V_i|/n}\right \rfloor$.
All algorithms are repeated 5 times, and we show the average diversity scores and average running time in each experiment. We set $\epsilon=0.1$ in the grid search for $\algmmd{}$ and $\greedyflow{}$ and in $\silp{}$.
We used the faster version of \algmmd{} and ran \algdec{} three times.

\ptitle {Experimental results} Figures~\ref{fig:divsmall}--\ref{fig:timesmall} compare the diversity scores and the running times of \algmmd{}, \fairflow{}, \greedyflow{} and \silp{} for increasing values of $k$.
For this experiment in particular, we sampled 1000 items uniformly at random from the full datasets, similar to the experimental setup in \cite{Yanhao1}.
Observe that \algmmd{} always outperforms the other polynomial time algorithms \fairflow{} and \greedyflow{} on the diversity score, and is often faster as well. Note that this is partially explainable by tighter constraints; however in our experiments \algmmd{} also outperformed the other polynomial time algorithms when used with the same constraints. 
\algmmd{} is competitive with the exponential time algorithm \silp{}.

Table~\ref{tab:speedfulldata} and Table~\ref{tab:performfulldata} show the performance on the full datasets for $k=15$.
\silp{} achieves the best performance, but \algmmd{} is a close second and outperforms \fairflow{} and \greedyflow{}.
We did not test on the \emph{CelebA} dataset, since the full data was not present in the repository of \cite{Yanhao1}. \greedyflow{} did not provide a solution within 24 hours on the \emph{Census} data set.

To compare the scalability of the different algorithms, we generate synthetic datasets.
Using the implementation of \cite{Yanhao1} we generate ten two-dimensional Gaussian isotropic clouds with random centers in $[-10,10]^2$ and identity covariance matrices. We set $k=30$. First, we fix $n=1000$ and increase $m$.
Then, we fix $m=3$ and increase the graph size $n$.
Figure~\ref{fig:syn} shows the results.
\greedyflow{} did not finish within 24 hours for $n=10^6$.
\algmmd{} scales worse for increasing $m$ when compared to the other baselines. This is mainly due to two reasons. First, our prepossessing step is repeated for each optimal guess $\gamma$ (Section~\ref{sec:decomp}), while \fairflow{} only does a greedy preprocessing once to obtain $km$ nodes, independent of the guess.
Secondly, the CKR partitions from the decomposition phase (Section~\ref{sec:decomp}) 
require a BFS search for each of the $km$ nodes. This is rather costly for increasing values of $m$.
Nonetheless, \algmmd{} achieves better diversity scores than the other polynomial time approximation algorithms \fairflow{} and \greedyflow{}, and is slightly worse than $\silp{}$.
\silp{} is tailored for performance and scalability to large datasets (by using fast ILP solvers), but has exponential running time in the worst case.

\begin{table}[t]
  \caption{Real-world datasets used for evaluation.}
  \label{tab:datasets}
  \begin{tabular}{@{}lrrrr@{}}
    \toprule
    Data \cite{Yanhao1}  & labels & $m$ & $n$ & dim\\
    \midrule
    \emph{Adult (S)} & sex & 2 & 48\,482 & 6 \\
    \emph{Adult (R)} & race & 5 &  &  \\
    \emph{Adult (S+R)} & sex+race & 10 &  &  \\
    \hline
    \emph{CelebA (S)} & sex & 2 & 202\,599 & 10 \\
    \emph{CelebA (A)} & age & 2 &  & \\
    \emph{CelebA (S+A)} & sex+age & 4 &  & \\
    \hline
    \emph{Census (S)} & sex & 2 & 2\,426\,116 & 25 \\
    \emph{Census (A)} & age & 7 & &  \\
    \emph{Census (S+A)} & sex+age & 14 & &  \\
    \hline
    \emph{Twitter (S)} & sex & 3 & 18\,836 & 1,024 \\
  \bottomrule
\end{tabular}
\end{table}

\begin{table}[] \centering
\caption{Running time (s) of different algorithms on the full datasets from Table~\ref{tab:datasets} for fixed parameter value $k=15$.}
\label{tab:speedfulldata}
\tcbset{colframe=black!5, colback=black!5, size=fbox, on line}
\begin{tabular}{@{}lrrrr@{}}\toprule
Data & \fairflow{} & \silp{} & \greedyflow{} & \algmmd{} \\
\midrule
\emph{Adult (S)} & 12.88 & 18.07 & 117.93 & 14.93 \\
\emph{Adult (R)} & 6.85 & 14.53 & 50.48 & 19 \\
\emph{Adult (S+R)} & 11.48 & 21.48 & 44.6 & 45 \\
\emph{Census (S)} & 516.88 & 688 & - & 492.46\\
\emph{Census (A)} & 658.6 & 1130.19 & - & 1630.8 \\
\emph{Census (S+A)} & 892.26 & 843.28 & - & 2289.88 \\
\emph{Twitter (S)} & 29.71 & 58.23 & 175.31 & 14.61 \\
\bottomrule
\end{tabular}
\end{table}

\begin{table}[] \centering
\caption{Diversity scores of different algorithms on the full datasets from Table~\ref{tab:datasets} for fixed parameter value $k=15$.}
\label{tab:performfulldata}
\tcbset{colframe=black!5, colback=black!5, size=fbox, on line}
\begin{tabular}{@{}lrrrr@{}}\toprule
Data & \fairflow{} & \silp{} & \greedyflow{} & \algmmd{} \\
\midrule
\emph{Adult (S)} & 3.63 & 5.93 & 1.57 & 5.18 \\
\emph{Adult (R)} & 1.67 & 5.49 & 1.13 & 3.96 \\
\emph{Adult (S+R)} & 0.96 & 5.25 & 0.76 & 2.94 \\
\emph{Census (S)} & 9.64 & 13.30 & - & 13.11\\
\emph{Census (A)} & 5.57 & 13.30 & - & 11.18 \\
\emph{Census (S+A)} & 1.01 & 13.38 & - & 9.11 \\
\emph{Twitter (S)} & 21.25 & 27.53 & 22.05 & 26.08 \\
\bottomrule
\end{tabular}
\end{table}

\begin{figure*}

\setlength{\figwidth}{3.9cm}
\setlength{\figheight}{3.3cm}
\setlength{\tabcolsep}{0pt}

\begin{tabular}{rrrrr}

\begin{tikzpicture}
\begin{axis}[xlabel={$k$}, ylabel= {$\diver{S}$},
    width = \figwidth,
    height = \figheight,
    title = {Adult (S)},
    ymin = 0.8,
    ymax = 7,
    scaled y ticks = false,
    cycle list name=yaf,
    yticklabel style={/pgf/number format/fixed},
    no markers,
    xtick = {10, 30, 50},
    ytick = {1, 3, 5, 7},
    legend columns = 1,
    legend style = {row sep=-2pt, inner sep=0pt},
    legend entries = {\labflow, \labilp, \labgrd, \labmmd}
    ]

\addplot table[x index = 0, y index = 1, header = false, col sep=comma] {results/div_Adults_C2.csv};
\addplot table[x index = 0, y index = 2, header = false, col sep=comma] {results/div_Adults_C2.csv};
\addplot table[x index = 0, y index = 3, header = false, col sep=comma] {results/div_Adults_C2.csv};
\addplot table[x index = 0, y index = 4, header = false, col sep=comma] {results/div_Adults_C2.csv};

\pgfplotsextra{\yafdrawaxis{5}{50}{1}{7}}
\end{axis}
\end{tikzpicture} &

\begin{tikzpicture}
\begin{axis}[xlabel={$k$}, ylabel= {$\diver{S}$},
    width = \figwidth,
    height = \figheight,
    title = {Adult (R)},
    ymin = 0,
    scaled y ticks = false,
    cycle list name=yaf,
    yticklabel style={/pgf/number format/fixed},
    no markers,
    xtick = {10, 30, 50},
    ]
\addplot table[x index = 0, y index = 1, header = false, col sep=comma] {results/div_Adults_C5.csv};
\addplot table[x index = 0, y index = 2, header = false, col sep=comma] {results/div_Adults_C5.csv};
\addplot table[x index = 0, y index = 3, header = false, col sep=comma] {results/div_Adults_C5.csv};
\addplot table[x index = 0, y index = 4, header = false, col sep=comma] {results/div_Adults_C5.csv};

\pgfplotsextra{\yafdrawaxis{5}{50}{0}{4.4}}
\end{axis}
\end{tikzpicture} &

\begin{tikzpicture}
\begin{axis}[xlabel={$k$}, ylabel= {$\diver{S}$},
    width = \figwidth,
    height = \figheight,
    title = {Adult (S+R)},
    ymin = 0,
    ymax = 3,
    scaled y ticks = false,
    cycle list name=yaf,
    yticklabel style={/pgf/number format/fixed},
    no markers,
    xtick = {10, 30, 50},
    ]
\addplot table[x index = 0, y index = 1, header = false, col sep=comma] {results/div_Adults_C10.csv};
\addplot table[x index = 0, y index = 2, header = false, col sep=comma] {results/div_Adults_C10.csv};
\addplot table[x index = 0, y index = 3, header = false, col sep=comma] {results/div_Adults_C10.csv};
\addplot table[x index = 0, y index = 4, header = false, col sep=comma] {results/div_Adults_C10.csv};

\pgfplotsextra{\yafdrawaxis{10}{50}{0}{3}}
\end{axis}
\end{tikzpicture} &

\begin{tikzpicture}
\begin{axis}[xlabel={$k$}, ylabel= {$\diver{S}$},
    width = \figwidth,
    height = \figheight,
    title = {Celeb (S)},
    scaled y ticks = false,
    cycle list name=yaf,
    yticklabel style={/pgf/number format/fixed},
    no markers,
    xtick = {10, 30, 50},
    ]
\addplot table[x index = 0, y index = 1, header = false, col sep=comma] {results/div_CelabA_C2.csv};
\addplot table[x index = 0, y index = 2, header = false, col sep=comma] {results/div_CelabA_C2.csv};
\addplot table[x index = 0, y index = 3, header = false, col sep=comma] {results/div_CelabA_C2.csv};
\addplot table[x index = 0, y index = 4, header = false, col sep=comma] {results/div_CelabA_C2.csv};

\pgfplotsextra{\yafdrawaxis{5}{50}{920}{2700}}
\end{axis}
\end{tikzpicture} &

\begin{tikzpicture}
\begin{axis}[xlabel={$k$}, ylabel= {$\diver{S}$},
    width = \figwidth,
    height = \figheight,
    title = {Celeb (A)},
    scaled y ticks = false,
    cycle list name=yaf,
    yticklabel style={/pgf/number format/fixed},
    no markers,
    xtick = {10, 30, 50},
    ]
\addplot table[x index = 0, y index = 1, header = false, col sep=comma] {results/div_CelabA_C2b.csv};
\addplot table[x index = 0, y index = 2, header = false, col sep=comma] {results/div_CelabA_C2b.csv};
\addplot table[x index = 0, y index = 3, header = false, col sep=comma] {results/div_CelabA_C2b.csv};
\addplot table[x index = 0, y index = 4, header = false, col sep=comma] {results/div_CelabA_C2b.csv};

\pgfplotsextra{\yafdrawaxis{5}{50}{880}{2600}}
\end{axis}
\end{tikzpicture} \\

\begin{tikzpicture}
\begin{axis}[xlabel={$k$}, ylabel= {$\diver{S}$},
    width = \figwidth,
    height = \figheight,
    title = {Celab (S+A)},
    scaled y ticks = false,
    cycle list name=yaf,
    yticklabel style={/pgf/number format/fixed},
    no markers,
    xtick = {10, 30, 50},
    ]
\addplot table[x index = 0, y index = 1, header = false, col sep=comma] {results/div_CelabA_C4.csv};
\addplot table[x index = 0, y index = 2, header = false, col sep=comma] {results/div_CelabA_C4.csv};
\addplot table[x index = 0, y index = 3, header = false, col sep=comma] {results/div_CelabA_C4.csv};
\addplot table[x index = 0, y index = 4, header = false, col sep=comma] {results/div_CelabA_C4.csv};

\pgfplotsextra{\yafdrawaxis{5}{50}{777}{2600}}
\end{axis}
\end{tikzpicture} &

\begin{tikzpicture}
\begin{axis}[xlabel={$k$}, ylabel= {$\diver{S}$},
    width = \figwidth,
    height = \figheight,
    title = {Census (S)},
    ymin = 4,
    ymax = 15,
    scaled y ticks = false,
    cycle list name=yaf,
    yticklabel style={/pgf/number format/fixed},
    no markers,
    xtick = {10, 30, 50},
    ]
\addplot table[x index = 0, y index = 1, header = false, col sep=comma] {results/div_Census_C2.csv};
\addplot table[x index = 0, y index = 2, header = false, col sep=comma] {results/div_Census_C2.csv};
\addplot table[x index = 0, y index = 3, header = false, col sep=comma] {results/div_Census_C2.csv};
\addplot table[x index = 0, y index = 4, header = false, col sep=comma] {results/div_Census_C2.csv};

\pgfplotsextra{\yafdrawaxis{5}{50}{4}{15}}
\end{axis}
\end{tikzpicture} &

\begin{tikzpicture}
\begin{axis}[xlabel={$k$}, ylabel= {$\diver{S}$},
    width = \figwidth,
    height = \figheight,
    title = {Census (A)},
    ymin = 0,
    ymax = 12,
    scaled y ticks = false,
    cycle list name=yaf,
    yticklabel style={/pgf/number format/fixed},
    no markers,
    xtick = {10, 30, 50},
    ]
\addplot table[x index = 0, y index = 1, header = false, col sep=comma] {results/div_Census_C7.csv};
\addplot table[x index = 0, y index = 2, header = false, col sep=comma] {results/div_Census_C7.csv};
\addplot table[x index = 0, y index = 3, header = false, col sep=comma] {results/div_Census_C7.csv};
\addplot table[x index = 0, y index = 4, header = false, col sep=comma] {results/div_Census_C7.csv};

\pgfplotsextra{\yafdrawaxis{10}{50}{0}{12}}
\end{axis}
\end{tikzpicture} &

\begin{tikzpicture}
\begin{axis}[xlabel={$k$}, ylabel= {$\diver{S}$},
    width = \figwidth,
    height = \figheight,
    title = {Census (S+A)},
    xmin = 15,
    ymin = 0,
    ymax = 10,
    scaled y ticks = false,
    cycle list name=yaf,
    yticklabel style={/pgf/number format/fixed},
    no markers,
    xtick = {15, 32, 50},
    ]
\addplot table[x index = 0, y index = 1, header = false, col sep=comma] {results/div_Census_C14.csv};
\addplot table[x index = 0, y index = 2, header = false, col sep=comma] {results/div_Census_C14.csv};
\addplot table[x index = 0, y index = 3, header = false, col sep=comma] {results/div_Census_C14.csv};
\addplot table[x index = 0, y index = 4, header = false, col sep=comma] {results/div_Census_C14.csv};

\pgfplotsextra{\yafdrawaxis{15}{50}{0}{10}}
\end{axis}
\end{tikzpicture} &

\begin{tikzpicture}
\begin{axis}[xlabel={$k$}, ylabel= {$\diver{S}$},
    width = \figwidth,
    height = \figheight,
    title = {Twitter (S)},
    ymin = 15,
    ymax = 30,
    scaled y ticks = false,
    cycle list name=yaf,
    yticklabel style={/pgf/number format/fixed},
    no markers,
    xtick = {10, 30, 50},
    ]
\addplot table[x index = 0, y index = 1, header = false, col sep=comma] {results/div_Twitter_C3.csv};
\addplot table[x index = 0, y index = 2, header = false, col sep=comma] {results/div_Twitter_C3.csv};
\addplot table[x index = 0, y index = 3, header = false, col sep=comma] {results/div_Twitter_C3.csv};
\addplot table[x index = 0, y index = 4, header = false, col sep=comma] {results/div_Twitter_C3.csv};

\pgfplotsextra{\yafdrawaxis{5}{50}{15}{30}}
\end{axis}
\end{tikzpicture}

\end{tabular}

\caption{Diversity scores of the solutions of different algorithms versus the solution size $k$ on a sample of $n=1000$.\label{fig:divsmall}}

\end{figure*}
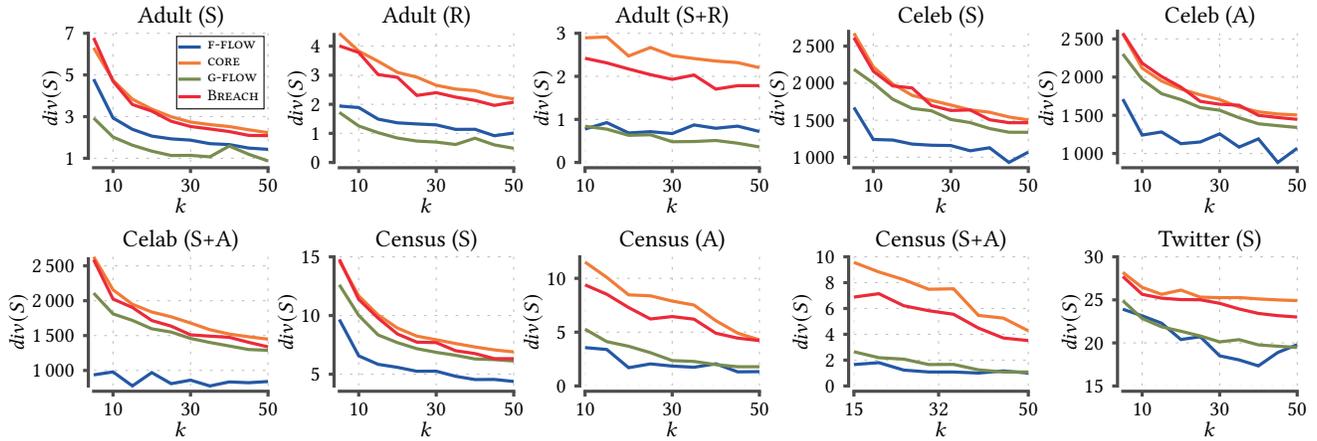

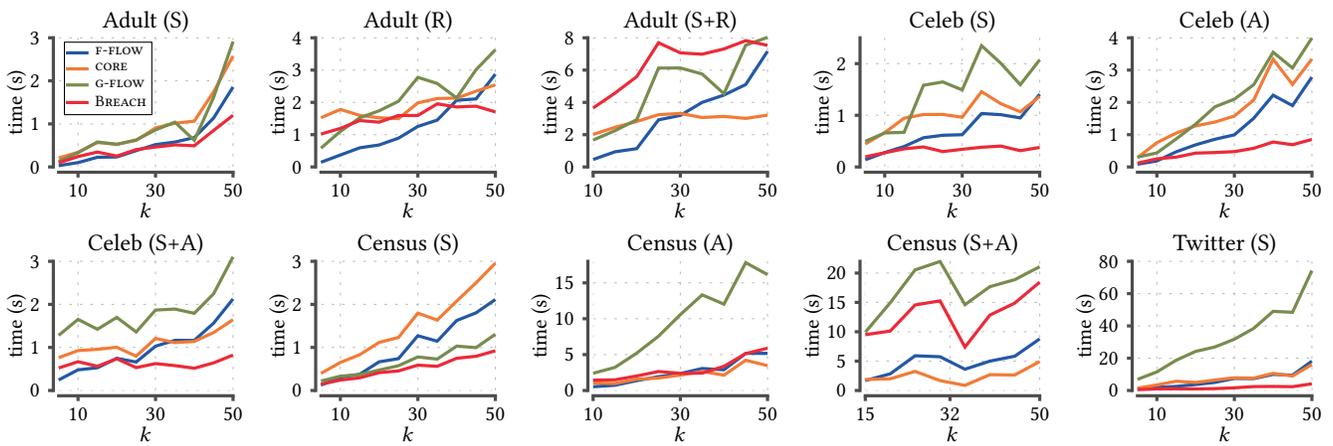
\begin{figure*}

\setlength{\figwidth}{3.9cm}
\setlength{\figheight}{3.3cm}

\begin{tabular}{rrrrr}

\begin{tikzpicture}
\begin{axis}[xlabel={$k$}, ylabel= {time (s)},
    width = \figwidth,
    height = \figheight,
    title = {Adult (S)},
    ymin = 0,
    ymax = 3,
    scaled y ticks = false,
    cycle list name=yaf,
    yticklabel style={/pgf/number format/fixed},
    no markers,
    xtick = {10, 30, 50},
    legend columns = 1,
    legend style = {row sep=-2pt, inner sep=0pt},
    legend pos = {north west},
    legend entries = {\labflow, \labilp, \labgrd, \labmmd}
    ]

\addplot table[x index = 0, y index = 1, header = false, col sep=comma] {results/time_Adults_C2.csv};
\addplot table[x index = 0, y index = 2, header = false, col sep=comma] {results/time_Adults_C2.csv};
\addplot table[x index = 0, y index = 3, header = false, col sep=comma] {results/time_Adults_C2.csv};
\addplot table[x index = 0, y index = 4, header = false, col sep=comma] {results/time_Adults_C2.csv};

\pgfplotsextra{\yafdrawaxis{5}{50}{0}{3}}
\end{axis}
\end{tikzpicture} &

\begin{tikzpicture}
\begin{axis}[xlabel={$k$}, ylabel= {time (s)},
    width = \figwidth,
    height = \figheight,
    title = {Adult (R)},
    ymin = 0,
    ymax = 4,
    scaled y ticks = false,
    cycle list name=yaf,
    yticklabel style={/pgf/number format/fixed},
    no markers,
    xtick = {10, 30, 50},
    ]
\addplot table[x index = 0, y index = 1, header = false, col sep=comma] {results/time_Adults_C5.csv};
\addplot table[x index = 0, y index = 2, header = false, col sep=comma] {results/time_Adults_C5.csv};
\addplot table[x index = 0, y index = 3, header = false, col sep=comma] {results/time_Adults_C5.csv};
\addplot table[x index = 0, y index = 4, header = false, col sep=comma] {results/time_Adults_C5.csv};

\pgfplotsextra{\yafdrawaxis{5}{50}{0.}{4}}
\end{axis}
\end{tikzpicture} &

\begin{tikzpicture}
\begin{axis}[xlabel={$k$}, ylabel= {time (s)},
    width = \figwidth,
    height = \figheight,
    title = {Adult (S+R)},
    ymin = 0,
    ymax = 8,
    scaled y ticks = false,
    cycle list name=yaf,
    yticklabel style={/pgf/number format/fixed},
    no markers,
    xtick = {10, 30, 50},
    ]
\addplot table[x index = 0, y index = 1, header = false, col sep=comma] {results/time_Adults_C10.csv};
\addplot table[x index = 0, y index = 2, header = false, col sep=comma] {results/time_Adults_C10.csv};
\addplot table[x index = 0, y index = 3, header = false, col sep=comma] {results/time_Adults_C10.csv};
\addplot table[x index = 0, y index = 4, header = false, col sep=comma] {results/time_Adults_C10.csv};

\pgfplotsextra{\yafdrawaxis{10}{50}{0}{8}}
\end{axis}
\end{tikzpicture} &

\begin{tikzpicture}
\begin{axis}[xlabel={$k$}, ylabel= {time (s)},
    width = \figwidth,
    height = \figheight,
    title = {Celeb (S)},
    ymin = 0,
    ymax = 2.5,
    scaled y ticks = false,
    cycle list name=yaf,
    yticklabel style={/pgf/number format/fixed},
    no markers,
    xtick = {10, 30, 50},
    ]
\addplot table[x index = 0, y index = 1, header = false, col sep=comma] {results/time_CelabA_C2.csv};
\addplot table[x index = 0, y index = 2, header = false, col sep=comma] {results/time_CelabA_C2.csv};
\addplot table[x index = 0, y index = 3, header = false, col sep=comma] {results/time_CelabA_C2.csv};
\addplot table[x index = 0, y index = 4, header = false, col sep=comma] {results/time_CelabA_C2.csv};

\pgfplotsextra{\yafdrawaxis{5}{50}{0}{2.5}}
\end{axis}
\end{tikzpicture} &

\begin{tikzpicture}
\begin{axis}[xlabel={$k$}, ylabel= {time (s)},
    width = \figwidth,
    height = \figheight,
    title = {Celeb (A)},
    ymin = 0,
    ymax = 4,
    scaled y ticks = false,
    cycle list name=yaf,
    yticklabel style={/pgf/number format/fixed},
    no markers,
    xtick = {10, 30, 50},
    ]
\addplot table[x index = 0, y index = 1, header = false, col sep=comma] {results/time_CelabA_C2b.csv};
\addplot table[x index = 0, y index = 2, header = false, col sep=comma] {results/time_CelabA_C2b.csv};
\addplot table[x index = 0, y index = 3, header = false, col sep=comma] {results/time_CelabA_C2b.csv};
\addplot table[x index = 0, y index = 4, header = false, col sep=comma] {results/time_CelabA_C2b.csv};

\pgfplotsextra{\yafdrawaxis{5}{50}{0}{4}}
\end{axis}
\end{tikzpicture} \\

\begin{tikzpicture}
\begin{axis}[xlabel={$k$}, ylabel= {time (s)},
    width = \figwidth,
    height = \figheight,
    title = {Celeb (S+A)},
    ymin = 0,
    ymax = 3,
    scaled y ticks = false,
    cycle list name=yaf,
    yticklabel style={/pgf/number format/fixed},
    no markers,
    xtick = {10, 30, 50},
    ]
\addplot table[x index = 0, y index = 1, header = false, col sep=comma] {results/time_CelabA_C4.csv};
\addplot table[x index = 0, y index = 2, header = false, col sep=comma] {results/time_CelabA_C4.csv};
\addplot table[x index = 0, y index = 3, header = false, col sep=comma] {results/time_CelabA_C4.csv};
\addplot table[x index = 0, y index = 4, header = false, col sep=comma] {results/time_CelabA_C4.csv};

\pgfplotsextra{\yafdrawaxis{5}{50}{0}{3}}
\end{axis}
\end{tikzpicture} &

\begin{tikzpicture}
\begin{axis}[xlabel={$k$}, ylabel= {time (s)},
    width = \figwidth,
    height = \figheight,
    title = {Census (S)},
    ymin = 0,
    ymax = 3,
    scaled y ticks = false,
    cycle list name=yaf,
    yticklabel style={/pgf/number format/fixed},
    no markers,
    xtick = {10, 30, 50},
    ]
\addplot table[x index = 0, y index = 1, header = false, col sep=comma] {results/time_Census_C2.csv};
\addplot table[x index = 0, y index = 2, header = false, col sep=comma] {results/time_Census_C2.csv};
\addplot table[x index = 0, y index = 3, header = false, col sep=comma] {results/time_Census_C2.csv};
\addplot table[x index = 0, y index = 4, header = false, col sep=comma] {results/time_Census_C2.csv};

\pgfplotsextra{\yafdrawaxis{5}{50}{0}{3}}
\end{axis}
\end{tikzpicture} &

\begin{tikzpicture}
\begin{axis}[xlabel={$k$}, ylabel= {time (s)},
    width = \figwidth,
    height = \figheight,
    title = {Census (A)},
    ymin = 0,
    ymax = 18,
    scaled y ticks = false,
    cycle list name=yaf,
    yticklabel style={/pgf/number format/fixed},
    no markers,
    xtick = {10, 30, 50},
    ]
\addplot table[x index = 0, y index = 1, header = false, col sep=comma] {results/time_Census_C7.csv};
\addplot table[x index = 0, y index = 2, header = false, col sep=comma] {results/time_Census_C7.csv};
\addplot table[x index = 0, y index = 3, header = false, col sep=comma] {results/time_Census_C7.csv};
\addplot table[x index = 0, y index = 4, header = false, col sep=comma] {results/time_Census_C7.csv};

\pgfplotsextra{\yafdrawaxis{10}{50}{0}{18}}
\end{axis}
\end{tikzpicture} &

\begin{tikzpicture}
\begin{axis}[xlabel={$k$}, ylabel= {time (s)},
    width = \figwidth,
    height = \figheight,
    title = {Census (S+A)},
    xmin = 15,
    ymin = 0,
    ymax = 22,
    scaled y ticks = false,
    cycle list name=yaf,
    yticklabel style={/pgf/number format/fixed},
    no markers,
    xtick = {15, 32, 50},
    ]
\addplot table[x index = 0, y index = 1, header = false, col sep=comma] {results/time_Census_C14.csv};
\addplot table[x index = 0, y index = 2, header = false, col sep=comma] {results/time_Census_C14.csv};
\addplot table[x index = 0, y index = 3, header = false, col sep=comma] {results/time_Census_C14.csv};
\addplot table[x index = 0, y index = 4, header = false, col sep=comma] {results/time_Census_C14.csv};

\pgfplotsextra{\yafdrawaxis{15}{50}{0}{22}}
\end{axis}
\end{tikzpicture} &

\begin{tikzpicture}
\begin{axis}[xlabel={$k$}, ylabel= {time (s)},
    width = \figwidth,
    height = \figheight,
    title = {Twitter (S)},
    ymin = 0,
    ymax = 80,
    scaled y ticks = false,
    cycle list name=yaf,
    yticklabel style={/pgf/number format/fixed},
    no markers,
    xtick = {10, 30, 50},
    ]
\addplot table[x index = 0, y index = 1, header = false, col sep=comma] {results/time_Twitter_C3.csv};
\addplot table[x index = 0, y index = 2, header = false, col sep=comma] {results/time_Twitter_C3.csv};
\addplot table[x index = 0, y index = 3, header = false, col sep=comma] {results/time_Twitter_C3.csv};
\addplot table[x index = 0, y index = 4, header = false, col sep=comma] {results/time_Twitter_C3.csv};

\pgfplotsextra{\yafdrawaxis{5}{50}{0}{80}}
\end{axis}
\end{tikzpicture}

\end{tabular}

\caption{Running time of different algorithms versus the solution size $k$ on a sample of $n=1000$.\label{fig:timesmall}}

\end{figure*}

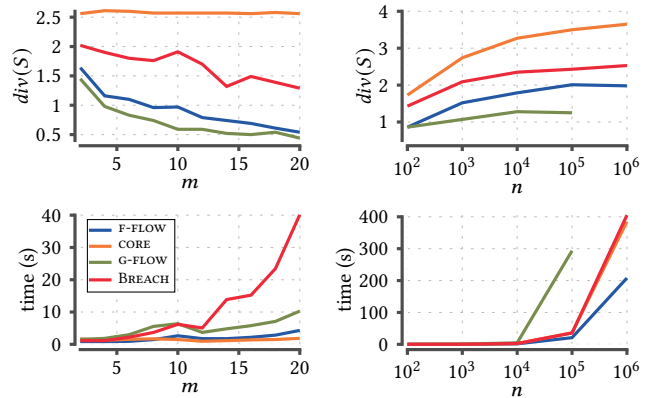
\begin{figure}

\setlength{\figwidth}{4.5cm}
\setlength{\figheight}{3.3cm}

\begin{tabular}{rr}

\begin{tikzpicture}[baseline=0]
\begin{axis}[xlabel={$m$}, ylabel= {$\diver{S}$},
    width = \figwidth,
    height = \figheight,
    ymin = 0.4,
    ymax = 2.6,
    scaled y ticks = false,
    cycle list name=yaf,
    yticklabel style={/pgf/number format/fixed},
    no markers,
    ]

\addplot table[x index = 0, y index = 1, header = false, col sep=comma] {results/div_vary_c.csv};
\addplot table[x index = 0, y index = 2, header = false, col sep=comma] {results/div_vary_c.csv};
\addplot table[x index = 0, y index = 3, header = false, col sep=comma] {results/div_vary_c.csv};
\addplot table[x index = 0, y index = 4, header = false, col sep=comma] {results/div_vary_c.csv};

\pgfplotsextra{\yafdrawaxis{2}{20}{0.4}{2.6}}
\end{axis}
\end{tikzpicture} &

\begin{tikzpicture}[baseline=0]
\begin{axis}[xlabel={$n$}, ylabel= {$\diver{S}$},
    width = \figwidth,
    height = \figheight,
    ymin = 0.5,
    ymax = 4,
    scaled y ticks = false,
    cycle list name=yaf,
    yticklabel style={/pgf/number format/fixed},
    no markers,
    xtick = {2, 3, 4, 5, 6},
    xticklabels = {$10^2$, $10^3$, $10^4$, $10^5$, $10^6$},
    ]

\addplot table[x index = 0, y index = 1, header = false, col sep=comma] {results/div_vary_n.csv};
\addplot table[x index = 0, y index = 2, header = false, col sep=comma] {results/div_vary_n.csv};
\addplot table[x index = 0, y index = 3, header = false, col sep=comma] {results/div_vary_n.csv};
\addplot table[x index = 0, y index = 4, header = false, col sep=comma] {results/div_vary_n.csv};

\pgfplotsextra{\yafdrawaxis{2}{6}{0.5}{4}}
\end{axis}
\end{tikzpicture} \\
\begin{tikzpicture}[baseline=0]
\begin{axis}[xlabel={$m$}, ylabel= {time (s)},
    width = \figwidth,
    height = \figheight,
    ymin = 0,
    ymax = 40,
    scaled y ticks = false,
    cycle list name=yaf,
    yticklabel style={/pgf/number format/fixed},
    no markers,
    legend pos = {north west},
    legend columns = 1,
    legend style = {row sep=-2pt, inner sep=0pt},
    legend entries = {\labflow, \labilp, \labgrd, \labmmd}
    ]

\addplot table[x index = 0, y index = 1, header = false, col sep=comma] {results/timing_vary_c.csv};
\addplot table[x index = 0, y index = 2, header = false, col sep=comma] {results/timing_vary_c.csv};
\addplot table[x index = 0, y index = 3, header = false, col sep=comma] {results/timing_vary_c.csv};
\addplot table[x index = 0, y index = 4, header = false, col sep=comma] {results/timing_vary_c.csv};

\pgfplotsextra{\yafdrawaxis{2}{20}{0}{40}}
\end{axis}
\end{tikzpicture} &

\begin{tikzpicture}[baseline=0]
\begin{axis}[xlabel={$n$}, ylabel= {time (s)},
    width = \figwidth,
    height = \figheight,
    scaled y ticks = false,
    cycle list name=yaf,
    yticklabel style={/pgf/number format/fixed},
    no markers,
    xtick = {2, 3, 4, 5, 6},
    xticklabels = {$10^2$, $10^3$, $10^4$, $10^5$, $10^6$},
    ]

\addplot table[x index = 0, y index = 1, header = false, col sep=comma] {results/timing_vary_n.csv};
\addplot table[x index = 0, y index = 2, header = false, col sep=comma] {results/timing_vary_n.csv};
\addplot table[x index = 0, y index = 3, header = false, col sep=comma] {results/timing_vary_n.csv};
\addplot table[x index = 0, y index = 4, header = false, col sep=comma] {results/timing_vary_n.csv};

\pgfplotsextra{\yafdrawaxis{2}{6}{0}{400}}
\end{axis}
\end{tikzpicture}

\end{tabular}

\caption{Diversity scores and the running times of the algorithms as a function of $m$ and $n$ on synthetic graphs.\label{fig:syn}}

\end{figure}
\section{Conclusions}
We proposed a novel randomized approximation algorithm for \FMMD called \algmmd{}.
Our algorithm uses random CKR partitions to find well-separated clusters and a novel maximum flow procedure to extract feasible solutions.
We theoretically prove our algorithm outperforms existing methods in a regime of few representatives per color, and also give a guarantee for the general problem.
In practice, our algorithm outperforms all other polynomial time approximation algorithms, while achieving slightly worse diversity score than the best exponential time algorithm. 
A direction for future work is to further narrow down the gap between the $\frac{1}{2}+\epsilon$ inapproxabimility lower bound and the approximation ratios from Table~\ref{table:summaryresults}.

%


\begin{acks}
This research is supported by the Academy of Finland project MALSOME (343045) and by the Helsinki Institute for Information Technology (HIIT).
\end{acks}

\bibliographystyle{ACM-Reference-Format}
\bibliography{kdd_2025}



\end{document}